\documentclass[12pt,a4paper]{article}
\usepackage[utf8]{inputenc}
\usepackage[english]{babel}
\usepackage{amsmath, amssymb, amsfonts}
\usepackage{amsthm} % ← Esto es necesario para los entornos de teoremas
\usepackage{physics}
\usepackage[T1]{fontenc}
\usepackage{lmodern}
\usepackage[outline]{contour}
\usepackage{graphicx}
\usepackage{cite}
\usepackage[colorlinks=true, linkcolor=blue, citecolor=red, urlcolor=blue]{hyperref}
\usepackage{hyperref}
\usepackage{tikz}
\usepackage{geometry}
\usepackage{dsfont}
\usepackage{float}
\usepackage{cleveref}
\usepackage{authblk}

\theoremstyle{plain}
\newtheorem{theorem}{Theorem}[section]
\newtheorem{lemma}[theorem]{Lemma}
\newtheorem{proposition}[theorem]{Proposition}
\theoremstyle{definition}
\newtheorem{definition}[theorem]{Definition}

\theoremstyle{remark}
\newtheorem{remark}[theorem]{Remark}
\author{Julio Cesar Jaramillo Quiceno\thanks{jcjaramilloq@unal.edu.co}}
\affil{
 \textsuperscript{1}Departamento de Física, Universidad Nacional de Colombia, Edificio Yu Takeuchi, Bogotá
}
\date{}
\title{$q$-Heisenberg Algebra in $\otimes^{2}-$Tensor Space}
\begin{document}
\maketitle
\begin{abstract}
In this paper, we introduce the $q$-Heisenberg algebra in the tensor product space $\otimes^2$. We establish its algebraic properties and provide applications to the theory of non-monogenic functions. Our results extend known constructions in $q$-deformed algebras and offer new insights into functional analysis in non-commutative settings.
\end{abstract}

\section{Introduction}
\subsection{Preliminaries}
\subsubsection{$\otimes^{2}$-tensor space}\label{tensor-2}
Building on the foundational work of Vakarchuk \cite{Vakarchuk2005}, we introduce and study the $q$-Heisenberg algebra within the framework of the $\otimes^{2}$-tensor space. This construction generalizes the conventional Heisenberg algebra, a mathematical structure central to quantum mechanics, incorporating a deformation parameter $q$. This parameter introduces a richer algebraic framework, allowing the description of more general quantum systems and providing a pathway to explore non-commutative phenomena, which are essential in modern theoretical physics and mathematics. The $\otimes^{2}$-tensor space provides a natural and powerful setting for this investigation. It allows for a clear and systematic representation of the algebraic relations and symmetries inherent in the $q$-deformed Heisenberg algebra. 
 \begin{definition}\rm{\cite{Nielsen-Chuang2010}}
In a formal setting, given a Hilbert space $\mathcal{H}$ that represents the state space of a quantum system, the \textit{\textbf{quantum $\otimes^{2}$-tensor space}}, denoted as $\mathcal{H} \otimes \mathcal{H}$, is the space of all possible tensor products of elements from $\mathcal{H}$ with themselves. This space can be understood as follows:
\begin{enumerate}
    \item[(i)] The tensor product operation $\otimes$ combines two vectors from $\mathcal{H}$ into a new vector in $\mathcal{H} \otimes \mathcal{H}$. If $|\psi\rangle, |\phi\rangle \in \mathcal{H}$, then the tensor product $|\psi\rangle \otimes |\phi\rangle$ represents a composite quantum state encompassing the information of both $|\psi\rangle$ and $|\phi\rangle$.
    \item[(ii)]  The resulting space $\mathcal{H} \otimes \mathcal{H}$ retains several properties crucial for quantum mechanics, including:
    \begin{itemize}
        \item[(ii.a)]  \textbf{Dimension}: If $\mathcal{H}$ has dimension $n$, then $\mathcal{H} \otimes \mathcal{H}$ has dimension $n^2$.
        \item[(ii.b)]  \textbf{Entanglement}: The states in this tensor space can exhibit entanglement, which is a fundamental aspect of quantum mechanics. For example, the state $\frac{1}{\sqrt{2}} (|\psi_1\rangle \otimes |\phi_2\rangle + |\psi_2\rangle \otimes |\phi_1\rangle )$ cannot be factored into product states.
    \end{itemize}
    \end{enumerate}
 \end{definition}
 In quantum mechanics, the $\otimes^{2}$-tensor space plays a critical role in the description of composite systems and entangled states. In this work presents a concise description of the quantum $\otimes^{2}$-tensor space, including references to essential sources for further reading.

\begin{definition}\rm{\cite{Shankar1980}}
In quantum mechanics, if we consider two quantum systems $A$ and $B$, their respective state spaces are Hilbert spaces, commonly denoted as $\mathcal{H}_{A}$ and $\mathcal{H}_{B}$. The combined system $A \otimes B$ is represented by the tensor product
\begin{equation*}
\mathcal{H}_{A \otimes B} = \mathcal{H}_{A} \otimes \mathcal{H}_{B}.
\end{equation*}

This space is known as \textit{\textbf{the quantum tensor space}} or bipartite $\otimes^{2}$-tensor space.
\end{definition}

\begin{definition}\rm{\cite{Cohen-Tannoudji-Diu-Laloe2020}}
    If $A$ and $B$ have basis vectors $\{|a_{i}\rangle\}$ and $\{|b_{j}\rangle\}$, respectively, then the basis vectors of $\mathcal{H}_{A \otimes B}$ are of the form $\{|a_{i}\rangle \otimes |b_{j}\rangle\}$. If $\dim(\mathcal{H}_{A}) = m$ and $\dim(\mathcal{H}_{B}) = n$, then $\dim(\mathcal{H}_{A \otimes B}) = m \times n$.
\end{definition}
\begin{definition}\rm{\cite{Nielsen-Chuang2010}}
This quantum $\otimes^{2}$-tensor space allows for the existence of entangled states, where the states of $A$ and $B$ are correlated such that they cannot be expressed as a simple product state $|a\rangle \otimes |b\rangle$. 
\end{definition}

\begin{definition}\rm{\cite{Fauser2000, Yuri-Manin2018}}\label{manim-def}
    Let \( q \) be a nonzero complex number (or a parameter in a general algebraic setting). The \textit{\textbf{quantum plane}} is defined as the associative algebra generated by two variables \( \hat{x}_{j} \) and \( \hat{x}_{k} \), subject to the commutation relation:
\begin{equation}\label{manin-jk}
\hat{x}_{j}\hat{x}_{k} = q\hat{x}_{k}\hat{x}_{j}.
\end{equation}

This algebra represents a deformation of the classical plane, which is recovered when \( q = 1 \). Now, let us consider the matrix 
    \begin{equation}\label{matrix2}
      \mathcal{M}_{q}(2) = 
      \begin{bmatrix}
            a & b\\
            c & d
    \end{bmatrix}.
    \end{equation}
    \textit{\textbf{The quantum determinant}} $Det_{q}\mathcal{M}_{q}(2)$ can be written as $ad-qcb$.
On other hand, let us consider $V$ as a vector space. The dimension of $V$ is defined as the number of linearly independent elements in $V$. Consider the space $\otimes^{2} V $ of degree-two tensors and impose a condition on elements $\hat{x}_{j}$ and $\hat{x}_{k}$ that span $V$, that is, $V = \langle \hat{x}_{j}, \hat{x}_{k} \rangle$, such that
\begin{equation}\label{jk-xxtensor}
    \hat{x}_{k} \otimes \hat{x}_{j}= q\hat{x}_{j} \otimes \hat{x}_{k}.
\end{equation}

Typically, the tensor product symbol is omitted for simplicity. Next, we consider the endomorphisms of $V$ that preserve condition $\hat{x}_{k}\otimes \hat{x}_{j}=q\hat{x}_{j}\otimes \hat{x}_{k}$, where $\mathcal{M}_{q}$ is given by (\ref{matrix2}) and $\vec{x}, \vec{x}^{\prime} \in V$, such that:
\begin{equation}
    \mathcal{M}_{q}\vec{x} = \vec{x}^{\prime}.
\end{equation}

This implies the requirement:
\begin{equation}
    \hat{x}^{\prime}_{k} \otimes \hat{x}_{j}^{\prime} = q \hat{x}_{j}^{\prime}\otimes \hat{x}^{\prime}_{k}.
\end{equation}

The constraints imposed on the quantum plane $V$ lead to specific relations among the matrix elements of $\mathcal{M}_{q}$ with respect to the basis of $V$ and its dual $V^{*}$. These matrix elements are necessarily non-commutative. Naturally, we have $V^{*} = \langle \frac{\partial}{\partial x_{j}}, \frac{\partial}{\partial x_{k}} \rangle$, and the dual quantum plane satisfies the relation:
\begin{equation}\label{momentum-tensor}
    \frac{\partial}{\partial x_{j}} \otimes \frac{\partial}{\partial x_{k}} = q^{-1} \frac{\partial}{\partial x_{k}} \otimes \frac{\partial}{\partial x_{j}}.
\end{equation}

Thus, this space forms the canonical dual space of $V$.    
\end{definition}

\begin{definition}\rm{\cite{Jaramillo2024}}
    Let us consider the real  algebra $\mathcal{B}_{p}$ is a real vector space with $2^{p}$ basis elements $ \mathrm{e}_{1},\dotsc , \mathrm{e}_{2^{p}-1}$, and let $\left\lbrace \mathrm{e}_{1} ,\mathrm{e}_{2},\mathrm{e}_{3},\dotsc ,\mathrm{e}_{p}\right\rbrace$ be a basis of $\mathbb{R}^{p}$. The multiplication in $\mathcal{B}_{p}$ are given by the rules
\begin{align}\label{rel1+a}
\mathrm{e}_{j}\mathrm{e}_{k}+q_{jk}\mathrm{e}_{k}\mathrm{e}_{j} & = \  \delta_{jk}\\
\label{rel1a}\mathrm{e}_{j}\mathrm{e}_{k}+\mathrm{e}_{k}\mathrm{e}_{j} & = \  2(1+q_{jk})\quad  j , k = 1,2,\dotsc , p,
\end{align}
being $q_{jk}$ 
\begin{equation}\label{rel2b}
q_{jk} =
\begin{cases}
-1 & j\neq k,\\
0 & j=k
\end{cases}.
\end{equation}
\end{definition}

\begin{definition}\rm{\cite{Jaramillo2024, Krasnov1999}}
The real Clifford algebra $\mathcal{A}_{m}$ is a real vector space with $2^{m}$ basis elements $\mathbf{e}_{0}, \mathbf{e}_{1},\dotsc , \mathbf{e}_{2^{m}-1}$, defined by
\begin{align*}
\mathbf{e}_{0}\equiv e_{0}=1, \mathbf{e}_{1} & = \ e_{1},\dotsc ,\mathbf{e}_{m}=e_{m},\\
\mathbf{e}_{12}=e_{1}e_{2}, \mathbf{e}_{13}=e_{1}e_{3}, \dotsc , \mathbf{e}_{m-1,m} & = \ e_{m-1}e_{m},\dotsc , \mathbf{e}_{12...m}=e_{1}e_{2}\dotsc e_{m},
\end{align*}

and let $\left\lbrace \mathbf{e}_{0},\mathbf{e}_{1} ,\dotsc,\mathbf{e}_{12},\mathbf{e}_{13},\dotsc,\mathbf{e}_{m-1,m},\mathbf{e}_{12\dotsc m}\right\rbrace$ be a basis of $\mathbb{R}^{m}$. The multiplication in $\mathcal{A}_{m}$ is given by the  rule
\begin{equation}\label{clifford}
e_{\alpha}e_{\beta}+e_{\beta}e_{\alpha} = -2\delta_{\alpha\beta}e_{0}, \quad  \alpha , \beta = 1,2,\dotsc, m.
\end{equation}
\end{definition}

\begin{definition}\rm{\cite{Jaramillo2024, Jaramillo2024-momento, Krasnov1999, Brackx1982, Ryan1996}}
The generalization of Cauchy - Riemann operator is given by 
\begin{equation}\label{Cauchy-Riemann}
\mathcal{D}=\sum\limits_{\beta =0}^{m}\mathbf{e}_{\beta}\frac{\partial}{\partial x_{\beta}}=\frac{\partial}{\partial x_{0}}+\sum\limits_{\beta =1}^{m}\mathbf{e}_{\beta}\frac{\partial}{\partial x_{\beta}},
\end{equation}

the second term correspond to \textit{Dirac Operator},  which we will denote by $D$ and is defined as

\begin{equation}\label{explicit-Dirac-operator}
D=\sum\limits_{\beta =1}^{m}\mathbf{e}_{\beta}\frac{\partial}{\partial x_{\beta}}.
\end{equation}
\end{definition}

 \begin{definition}\rm{\cite{Jaramillo2024}}
Let $\mathrm{e}_{1} ,\mathrm{e}_{2}, \mathrm{e}_{3},...,\mathrm{e}_{p},$ be elements that satisfy (\ref{rel1+a}). The difference operator  $\textit{\textbf{D}}$ is defined as
\begin{equation}\label{op}
\textit{\textbf{D}}=\mathrm{e}_{j}\frac{\partial}{\partial x_{k}}+\mathrm{e}_{k}\frac{\partial}{\partial x_{j}},
\end{equation}

which is subject to
\begin{multline}
\mathrm{e}_{j}\frac{\partial}{\partial x_{k}}\mathrm{e}_{k}\frac{\partial}{\partial x_{j}}+\mathrm{e}_{k}\frac{\partial}{\partial x_{j}}\mathrm{e}_{j}\frac{\partial}{\partial x_{k}} 
=\\
-\mathrm{e}_{j}\frac{\partial}{\partial x_{k}}\left[	\left(\mathrm{e}_{j}\frac{\partial}{\partial x_{k}}\right)\delta_{jk}\right]-\mathrm{e}_{k}\frac{\partial}{\partial x_{j}}\left[\left(\mathrm{e}_{k}\frac{\partial}{\partial x_{j}}\right)\delta_{jk}\right]
-(1-q_{jk})\left(\frac{\partial^{2}}{\partial x^{2}_{k}}+\frac{\partial^{2}}{\partial x_{j}^{2}}\right).
\end{multline}
\end{definition}

\begin{definition}\rm{\cite{Vakarchuk2005}}\label{Heis-f}
    Let $\hat{x}_{j}, \hat{x}_{k}, \hat{p}_{j}, \hat{p}_{k}$ be the position and momentum operators, and let $f$ be a function dependent on the particle's coordinates. The deformed Heisenberg algebra is subject to the following relations: 
    \begin{equation}\label{def-heise}
        [\hat{x}_{j}, \hat{x}_{k}]=0, \quad [\hat{x}_{j}, \hat{p}_{k}]=i\hbar\delta_{j,k}f, \quad [\hat{p}_{j}, \hat{p}_{k}]=-i\hbar\left(\frac{\partial f}{\partial x_{j}}\hat{p}_{k}-\frac{\partial f}{\partial x_{k}}\hat{p}_j\right), \quad j,k=1,2,3.
    \end{equation}
    where $\hbar$ is a Planck's constant.
\end{definition}

This paper are organized as follows: in section 2, we present the deformed Heisenberg algebra inn the $\otimes^{2}$-tensor space. In the final section, we presents  the application to non-monogenic function.

\section{Deformed Heisenberg algebra in the $\otimes^{2}$-Tensor space}

\begin{definition}\label{Tensor-Heis}
Let $\hat{x}_{j}, \hat{x}_{k}, \hat{p}_{j}, \hat{p}_{k}$ be elements of the vector space $V$. Over the space $\otimes^{2}V$ the \textit{\textbf{deformed Heisenberg algebra in the \texorpdfstring{$\otimes^{2}$}{Lg}-tensor space}} is defined by the following relations:
\begin{align}\label{tensor-heis+1}
    \hat{x}_{j}\otimes\hat{x}_{k}-q^{-1}\hat{x}_{k}\otimes\hat{x}_{j} & = \ 0, \quad \hat{p}_{j}\otimes\hat{p}_{k}-q^{-1}\hat{p}_{k}\otimes\hat{p}_{j} = 0,\\
  \notag\hat{x}_{j}\otimes\hat{p}_{k}-q\hat{p}_{k}\otimes\hat{x}_{j} &  = \  i\hbar\delta_{jk}.
\end{align}
\end{definition}
The following definition is a particular case of Definition \ref{Tensor-Heis} in terms of a function $f$ that depends on the particle coordinates $x_{j}, x_{k}$ based of the Vakarchuk \cite{Vakarchuk2005}.
%By utilizing the properties of tensor products, we can better understand the relationship between algebraic structures and physical systems, particularly in contexts where non-commutativity plays a crucial role, such as in quantum field theory, string theory, and non-commutative geometry.

%%In this work, we explore the fundamental properties of the $q$-Heisenberg algebra, examining its mathematical structure, physical implications, and potential applications. Our analysis bridges the gap between abstract algebraic concepts and their concrete realizations in physical theories. By doing so, we aim to provide a comprehensive and accessible understanding of this algebra, opening new avenues for further research and applications in areas such as quantum information theory, condensed matter physics, and the study of quantum symmetries.

\begin{definition}\rm{\cite{Vakarchuk2005}}
    Let \( f \) be a function depending on the coordinates of the particle \( x_{j} \) and \( x_{k} \). From Definition \ref{Heis-f}, we define the version for the \textit{\textbf{\( q \)-Heisenberg algebra in the \( \otimes^{2} \)-space}} in terms of the function \( f \), generated by the operators \( \hat{x}_{j} \), \( \hat{x}_{k} \), \( \hat{p}_{j} \), and \( \hat{p}_{k} \), through the following relations:
    \begin{align}
        \label{R1} 
        \hat{x}_{j} \otimes \hat{x}_{k} &= q^{-1} \hat{x}_{k} \otimes \hat{x}_{j}, \\
        \hat{x}_{j} \otimes \hat{p}_{k} - q_{jk} \hat{p}_{k} \otimes \hat{x}_{j} &= -i\hbar \delta_{jk} f, \\
        \label{R2} 
        \hat{p}_{j} \otimes \hat{p}_{k} - q^{-1} \hat{p}_{k} \otimes \hat{p}_{j} &= -i\hbar \left( \frac{\partial f}{\partial x_{j}} \otimes \hat{p}_{k} - q^{-1} \frac{\partial f}{\partial x_{k}} \otimes \hat{p}_{j} \right).
    \end{align}
\end{definition}

\begin{lemma}
    If $f=1$ then the relations \eqref{tensor-heis+1} are obtained
\end{lemma}
\begin{proof}
   The proof is completed by demonstrating that $\frac{\partial}{\partial x_{j}}(1) = \frac{\partial}{\partial x_{k}}(1) = 0$. This leads to the relations (\ref{momentum-tensor}) and (\ref{jk-xxtensor}), which are obtained when \( f = 1 \). These results follow from the definitions of the momentum operators $\hat{p}_{j} = -i\hbar\frac{\partial}{\partial x_{j}}$ and $\hat{p}_{k} = -i\hbar\frac{\partial}{\partial x_{k}}$.
\end{proof}

\section{Application to non-monogenic functions}

In the following section, we will study this deformation of the Heisenberg algebra in tensor space, where the functions are defined within the Clifford algebra, specifically for the non-monogenic case.

\begin{proposition}\rm{\cite{Jaramillo2024}}\label{prop}
 Let $f$ be left non-monogenic function of the form $f=f_{j}(x_{j},x_{k})\mathbf{e}_{j}+f_{k}(x_{j},x_{k})\mathbf{e}_{k}$ (the Dirac operator $Df\neq 0$). For this case, the $q$- Heisenberg algebra in the $\otimes^{2}$-space is given by the following relations
 \begin{align}
     \label{ro1}\hat{x}\otimes\hat{x}_{k}& = \ q^{-1}\hat{x}_{k}\otimes\hat{x},\quad \hat{x}_{j}\otimes\hat{x} =  q^{-1}\hat{x}\otimes\hat{x}_{j},\\
     \label{ro2}\hat{x}\otimes\hat{p}_{k}-q_{jk}\hat{p}_{k}\otimes\hat{x}& = \ -i\hbar\delta_{j,k}f, \quad\hat{x}_{j}\otimes\hat{p}-q_{jk}\hat{p}\otimes\hat{x}_{j}=-i\hbar\delta_{j,k},\\
     \label{ro3}\hat{p}\otimes\hat{p}_{k}-q^{-1}\hat{p}_{k}\otimes\hat{p} & = \ -i\hbar\left(Df\otimes\hat{p}_{k}-q^{-1}\frac{\partial f}{\partial x_{k}}\otimes\hat{p}\right),\\
     \label{ro4}\hat{p}_{j}\otimes\hat{p}-q^{-1}\hat{p}\otimes\hat{p}_{j} & = \ -i\hbar \left(\frac{\partial f}{\partial x_{j}}\otimes\hat{p}-q^{-1}Df\otimes\hat{p}_{j}\right).
 \end{align}
\end{proposition}
\begin{proof}
    To derive the first relation in (\ref{ro1}), we multiply the first relation in (\ref{R1}) on the left by \(\textbf{e}_{j}\), obtaining
\begin{align*}
    \textbf{e}_{j} \hat{x}_{j} \otimes \hat{x}_{k} &= q^{-1} \textbf{e}_{j} \hat{x}_{k} \otimes \hat{x}_{j}, \\
    \textbf{e}_{j} \hat{x}_{j} \otimes \hat{x}_{k} &= q^{-1} \hat{x}_{k} \otimes \textbf{e}_{j} \hat{x}_{j}, \\
    \hat{x} \otimes \hat{x}_{k} &= q^{-1}\hat{x}_{k} \otimes \hat{x}.
\end{align*}

A similar procedure can be applied to obtain the second relation in (\ref{ro1}), this time multiplying on the right by \(\textbf{e}_{k}\). Now to obtain the firts relation in (\ref{ro2}), we consider $f=f_{j}(x)$ and we multiply on the left by $\mathbf{e}_{j}$ resulting in
\begin{align*}
    \mathbf{e}_{j}\hat{x}_{j}\otimes\hat{p}_{k}-q_{jk}\mathbf{e}_{j}\hat{p}_{k}\otimes\hat{x}_{j} & = \ -i\hbar\delta_{jk}f_{j}(x)\mathbf{e}_{j},\\
    \mathbf{e}_{j}\hat{x}_{j}\otimes\hat{p}_{k}-q_{jk}\hat{p}_{k}\otimes\mathbf{e}_{j}\hat{x}_{j} & = \ -i\hbar\delta_{jk}f_{j}(x)\mathbf{e}_{j},\\
    \hat{x}\otimes\hat{p}_{k}-q_{jk}\hat{p}_{k}\otimes\hat{x} & = \ -i\hbar\delta_{jk}f.
\end{align*}
To obtain the second relation in (\ref{ro2}), we apply the same procedure as above, this time multiplying on the right by \(\mathbf{e}_{k}\).  Now to obtain (\ref{ro3}) we consider the Dirac operator defined in (\ref{explicit-Dirac-operator}). Multiplying on the left-hand by $\mathbf{e}_{j}$ in (\ref{R2}) we have
\begin{align*}
 \mathbf{e}_{j}\hat{p}_{j}\otimes\hat{p}_{k}-q^{-1}\mathbf{e}_{j}\hat{p}_{k}\otimes\hat{p}_{j}  & = \ -i\hbar\left(\mathbf{e}_{j}\frac{\partial f}{\partial x_{j}}\otimes\hat{p}_{k}-q^{-1}\mathbf{e}_{j}\frac{\partial f}{\partial x_{k}}\otimes\hat{p}_{j}\right),\\
 \hat{p}\otimes\hat{p}_{k}-q^{-1}\hat{p}_{k}\otimes\hat{p} & = \ -i\hbar\left(Df\otimes\hat{p}_{k}-q^{-1}\frac{\partial f}{\partial x_{k}}\otimes\mathbf{e}_{j}\hat{p}_{j}\right),\\
 & = \ -i\hbar\left(Df\otimes\hat{p}_{k}-q^{-1}\frac{\partial f}{\partial x_{k}}\otimes\hat{p}\right).
\end{align*}
Finally for (\ref{ro4}), we follow the same procedure as before, but this time multiplying on the right by \(\mathbf{e}_{k}\), and noting that the Dirac operator can be expressed as \(D = \mathbf{e}_{k} \frac{\partial f}{\partial x_{k}}\).
\end{proof}
\begin{theorem}
For monogenic functions, i.e., functions satisfying $Df = 0$ in \eqref{ro3} and \eqref{ro4}, 
the expression \eqref{momentum-tensor} holds.
\end{theorem}
\begin{proof}
Assume $f = 0\mathbf{e}_{j} + 0\mathbf{e}_{k}$. Then, $Df = 0$ and 
\[
\frac{\partial f}{\partial x_{j}} = \frac{\partial f}{\partial x_{k}} = 0.
\]
Substituting these into \eqref{ro3} and \eqref{ro4} yields
\[
\hat{p} \otimes \hat{p}_{k} - q^{-1}\hat{p}_{k} \otimes \hat{p} = 0, \quad 
\hat{p}_{j} \otimes \hat{p} - q^{-1}\hat{p} \otimes \hat{p}_{j} = 0,
\]
and taking into account that $\hat{p}=-i\hbar\mathbf{e}_{j}\frac{\partial}{\partial x_{j}}$ and $\hat{p}=-i\hbar\mathbf{e}_{k}\frac{\partial}{\partial x_{k}}$, then we obtain \eqref{momentum-tensor} \qedhere.
\end{proof}

\begin{remark}\rm{\cite{Jaramillo2024}}
    The above treatment can be applied to right-non-monogenic functions, $fD\neq 0$.
\end{remark}

In the following section, we will propose the corresponding 
$q$-deformed Heisenberg algebra in the space $\otimes^{2}$.

\begin{proposition}\label{prop1}
The $q$-Heisenberg algebra is defined by the set of generators $\hat{\textit{\textbf{x}}},\hat{\textit{\textbf{p}}}$, where $\hat{\textit{\textbf{x}}}=\hat{x}_{j}\mathrm{e}_{k}+\hat{x}_{k}\mathrm{e}_{j},\hat{\textit{\textbf{p}}}=\hat{p}_{j}\mathrm{e}_{k}+\hat{p}_{k}\mathrm{e}_{j}$, and $f=f_{j}(x_{j},x_{k})\mathrm{e}_{k}+f_{k}(x_{j},x_{k})\mathrm{e}_{j}$, the difference operator \rm(\ref{op}), and is subject to following relations:
\begin{align}
    \label{t1a}\hat{\textit{\textbf{x}}}\otimes\hat{p}_{k}-q_{jk}\hat{p}_{k}\otimes\hat{\textit{\textbf{x}}} & = \ - i\hbar\delta_{jk}f,\\
    \label{t2+a}\hat{x}_{j}\otimes\hat{\textit{\textbf{p}}}-q_{jk}\hat{\textit{\textbf{p}}}\otimes\hat{x}_{j} & = \ - i\hbar\delta_{jk}f,\\
    \label{t3}\mathrm{e}_{k}\hat{p}_{j}\otimes\hat{p}_{k}-q^{-1}_{jk}\hat{p}_{k}\otimes\mathrm{e}_{k}\hat{p}_{j}  & = \ -i\hbar\left(\mathrm{e}_{k}\frac{\partial f}{\partial x_{j}}\otimes\hat{p}_{k}- q^{-1}_{jk}\mathrm{e}_{k}\frac{\partial f}{\partial x_{k}}\otimes\hat{p}_{j}\right),\\
    \label{t4}\hat{p}_{j}\otimes\mathrm{e}_{j}\hat{p}_{k}-q^{-1}_{jk}\mathrm{e}_{j}\hat{p}_{k}\otimes\hat{p}_{j} & = \ -i\hbar\left(\mathrm{e}_{j}\frac{\partial f}{\partial x_{j}}\otimes\hat{p}_{k}-q^{-1}_{jk}\mathrm{e}_{j}\frac{\partial f}{\partial x_{k}}\otimes\hat{p}_{j}\right).
\end{align}
\end{proposition}
\begin{proof}
    For (\ref{t1a}), we have that
    \begin{align*}
        i\hbar\delta_{jk}\mathrm{e}_{k}f_{j}& = \ \mathrm{e}_{k}\hat{x}_{j}\otimes\hat{p}_{k}-q_{jk}\hat{p}_{k}\otimes\mathrm{e}_{k}\hat{x}_{j},\\
        & = \ \mathrm{e}_{k}\hat{x}_{j}\otimes\hat{p}_{k}-q_{jk}\hat{p}_{k}\otimes\mathrm{e}_{k}\hat{x}_{j} +\mathrm{e}_{j}\hat{x}_{k}\otimes\hat{p}_{k}-\mathrm{e}_{j}\hat{x}_{k}\otimes\hat{p}_{k}-q_{jk}\hat{p}_{k}\otimes\mathrm{e}_{j}\hat{x}_{k}+q_{jk}\hat{p}_{k}\otimes\mathrm{e}_{j}\hat{x}_{k},\\& = \ (\mathrm{e}_{j}\hat{x}_{k}+\mathrm{e}_{k}\hat{x}_{j})\otimes\hat{p}_{k}-q_{jk}\hat{p}_{k}\otimes (\mathrm{e}_{j}\hat{x}_{k}+\mathrm{e}_{k}\hat{x}_{j})-\mathrm{e}_{j}\hat{x}_{k}\otimes\hat{p}_{k}+q_{jk}\hat{p}_{k}\otimes\mathrm{e}_{j}\hat{x}_{k} ,\\
        & = \ \hat{\textit{\textbf{x}}}\otimes\hat{p}_{k}-q_{jk}\hat{p}_{k}\otimes\hat{\textit{\textbf{x}}}-\mathrm{e}_{j}\hat{x}_{k}\otimes\hat{p}_{k}+q_{jk}\hat{p}_{k}\otimes\mathrm{e}_{j}\hat{x}_{k}.
  \end{align*}
  
  Taking into account that $-\mathrm{e}_{j}\hat{x}_{k}\otimes\hat{p}_{k}+q_{jk}\hat{p}_{k}\otimes\mathrm{e}_{j}\hat{x}_{k}=i\hbar\mathrm{e}_{j}f_{k}$, the above expression can be written of the following form
  \begin{equation*}
      \hat{\textit{\textbf{x}}}\otimes\hat{p}_{k}-q_{jk}\hat{p}_{k}\otimes\hat{\textit{\textbf{x}}}+i\hbar\delta_{jk}\mathrm{e}_{j}f_{k} = -i\hbar\delta_{jk}\mathrm{e}_{k}f_{j},
  \end{equation*}
  therefore we have
  \begin{equation*}
      \hat{\textit{\textbf{x}}}\otimes\hat{p}_{k}-q_{jk}\hat{p}_{k}\otimes\hat{\textit{\textbf{x}}}=-i\hbar\delta_{jk} f.
  \end{equation*}
  A similar procedure can be used to obtain (\ref{t2+a}). Now for (\ref{t3}) we have the following calculation
  \begin{align*}
      \mathrm{e}_{k}\hat{p}_{j}\otimes\hat{p}_{k}-q^{-1}_{jk}\hat{p}_{k}\otimes\mathrm{e}_{k}\hat{p}_{j} & = \ -i\hbar\left(\mathrm{e}_{k}\frac{\partial f}{\partial x_{j}}\otimes\hat{p}_{k}-q^{-1}_{jk}\mathrm{e}_{k}\frac{\partial f}{\partial x_{k}}\otimes\hat{p}_{j}\right),\\
      & = \ -i\hbar\mathrm{e}_{k}\frac{\partial f}{\partial x_{j}}\otimes\hat{p}_{k}+i\hbar q^{-1}_{jk}\mathrm{e}_{k}\frac{\partial f}{\partial x_{k}}\otimes\hat{p}_{j}.
      \end{align*}   
Now, for (\ref{t4}), we proceed in a similar manner as in the previous case, starting from the expression
\begin{equation*}
 \hat{p}_{j}\otimes\mathrm{e}_{j}\hat{p}_{k}-q^{-1}_{jk}\mathrm{e}_{j}\hat{p}_{k}\otimes\hat{p}_{j}=-i\hbar\left(\mathrm{e}_{j}\frac{\partial f}{\partial x_{j}}\otimes\hat{p}_{k}-q^{-1}_{jk}\mathrm{e}_{j}\frac{\partial f}{\partial x_{k}}\otimes\hat{p}_{j}\right).   
\end{equation*}
\end{proof}

The Propostion \ref{prop1}  plays a crucial role in the framework of noncommutative geometry and quantum algebra for several reasons:
\begin{itemize}
\item[(i)] \textit{\textbf{The extension of the classical Heisenberg algebra}}: By deforming the canonical commutation relations, this formulation allows the study of quantum systems on noncommutative spaces where the usual assumptions of coordinate-momentum duality no longer hold.
\item[(ii)] \textit{\textbf{Bridge to quantum geometry}}: The appearance of \( q \)-parameters and functionally dependent operators reflects an underlying geometric deformation, compatible with quantum groups and quantum planes, such as those appearing in Manin's formulation (see Definition \ref{manim-def}).
\item[(iii)]  \textit{\textbf{Foundations for noncommutative functional analysis}}: The modified commutators provide a natural starting point for constructing differential and integral calculus on \( q \)-deformed spaces, including Jackson derivatives and generalized Dirac operators.
\item[(iv)] \textit{\textbf{Relevance to high-energy physics and quantum gravity}}: Since the proposition embeds a minimal length scale through deformation parameters, it may offer insights into quantum field theories with ultraviolet regularization, and models of spacetime with discrete or quantum structure.
\item[(v)] \textit{\textbf{Structural compatibility with Clifford-type deformations}}: The tensorial formalism used in this proposition makes it compatible with further deformations, such as \( q \)-Clifford algebras, which are essential in describing spinorial structures and supersymmetric extensions in a noncommutative setting.
\end{itemize}
Therefore, Proposition \ref{prop1} lays the groundwork for a consistent and physically meaningful extension of the Heisenberg algebra, opening a pathway for future developments in algebraic methods of quantum physics and non-commutative geometry.

\bibliographystyle{plain} % o alpha, unsrt, ieeetr, etc.
\bibliography{myBiblib}
\end{document}